\documentclass[journal]{IEEEtran}
\usepackage{multirow}
\usepackage{enumitem}
\usepackage{array,booktabs}
\usepackage{tabularx, booktabs}
\usepackage{supertabular,booktabs}

\usepackage{amsfonts,amsmath,amsthm}       % blackboard math symbols
\usepackage{graphicx}
\usepackage{epsfig}
\usepackage{flushend}
\usepackage{tikz}
\usetikzlibrary{matrix,chains,positioning,decorations.pathreplacing,arrows}

\newtheorem{thm}{Theorem}

\newtheorem{cor}{Corollary}

\newtheorem{definition}{Definition}

\begin{document}

\title{SNIFF: Reverse Engineering of Neural Networks with Fault Attacks}

\author{Jakub Breier, Dirmanto Jap, Xiaolu Hou, Shivam Bhasin and Yang Liu
\thanks{J. Breier is with Silicon Austria Labs, TU-Graz SAL DES Lab and Graz University of Technology, Graz, Austria. E-mail: jbreier@jbreier.com}
\thanks{D. Jap and S. Bhasin are with Temasek Laboratories, Nanyang Technological University, Singapore.
E-mail: \{djap, sbhasin\}@ntu.edu.sg}
\thanks{X. Hou is with Faculty of Informatics and Information Technologies, Slovak University of Technology, Slovakia.
E-mail: houxiaolu.email@gmail.com}
\thanks{Y. Liu is with the School of Computer Science and Engineering, Nanyang Technological University, Singapore.
E-mail: yangliu@ntu.edu.sg}
\thanks{
This work has been supported in parts by the ``University SAL Labs'' initiative of Silicon Austria Labs (SAL) and its Austrian partner universities for applied fundamental research for electronic based systems. The authors acknowledge the support from the Singapore National Research Foundation (``SOCure'' grant NRF2018NCR-NCR002-0001 -- www.green-ic.org/socure).
This project has received funding from the European Union's Horizon 2020 Research and Innovation Programme under the Programme SASPRO 2 COFUND Marie Sklodowska-Curie grant agreement No. 945478.}
}

\markboth{IEEE Transactions on Reliability,~Vol.~??, No.~?, October~2020}%
{J. Breier, D. Jap, X. Hou, S. Bhasin, Y. Liu: SNIFF: Reverse Engineering of Neural Networks with Fault Attacks}
% {Shell \MakeLowercase{\textit{et al.}}: Bare Demo of IEEEtran.cls for Computer Society Journals}

\maketitle

\begin{abstract}
Neural networks have been shown to be vulnerable against fault injection attacks.
These attacks change the physical behavior of the device during the computation, resulting in a change of value that is currently being computed.
They can be realized by various techniques, ranging from clock/voltage glitching to application of lasers to rowhammer.
Previous works have mostly explored fault attacks for output misclassification, thus affecting the reliability of neural networks. In this paper we investigate the possibility to reverse engineer neural networks with fault attacks.
\underline{S}ig\underline{n} b\underline{i}t \underline{f}lip \underline{f}ault (SNIFF) attack enables the reverse engineering by changing the sign of intermediate values.
We develop the first exact extraction method on deep-layer feature extractor networks that provably allows the recovery of proprietary model parameters.
Our experiments with Keras library show that the precision error for the parameter recovery for the tested networks is less than $10^{-13}$ with the usage of 64-bit floats, which improves the current state of the art by 6 orders of magnitude.
\end{abstract}
% \fi
\IEEEpeerreviewmaketitle

\section{Introduction}
Neural networks form a basis for current artificial intelligence applications.
They were shown to be effective in domains that can provide large amount of labeled data to be able to learn the classification model with sufficient level of accuracy.
Because of this property, companies often protect their models as the cost of obtaining the data used to train them might be very high, while the availability of such data is limited.
Thus, having a classification model whose internal parameters are secret gives companies a competitive advantage.
It is therefore necessary to know the ways that enable reverse engineering of the models so that adequate protection could be applied.

Model stealing attacks (also called model extraction attacks) aim at retrieving the model parameters in a black-box settings~\cite{tramer2016stealing}.
In this setting, the attacker sends inputs to the network and observes the outputs.
Based on this information, she tries to reconstruct the model that has accuracy close to the original one.
In a similar fashion, it is possible to recover the hyperparameters of machine learning models in general~\cite{wang2018stealing}.

There are certain similarities when it comes to comparing the model stealing attacks with the key recovery attacks on cryptography.
Classical cryptanalysis works by querying the cryptosystem with inputs and observing the outputs.
This helps in getting the information about the secret key.
In the field of cryptography, researchers started observing the physical characteristics of the devices that perform the encryption to find the secret key more efficiently.
Similarly, it was shown that by causing errors during the cryptographic computation, the attacker can learn secret information~\cite{biham1997differential}.
We call these implementation-level attacks \textit{physical attacks} on cryptography.

Now, we can look into the emerging area concerned with physical attacks against neural networks.
It was shown earlier that side-channel attacks can be applied to neural networks to recover certain model parameters~\cite{batina2018csi}.
It was also shown that neural networks are vulnerable to fault injection attacks that change the intermediate values of the model during the computation, enabling misbehavior of the activation functions in the model~\cite{breier2018practical}.
As the fault might also be introduced through external factors, the reliability of the neural network implementations is becoming a growing concern. 
In some cases, a single fault occurring in a GPU could reduce the reliability of a CNN performance~\cite{DBLP:journals/tr/SantosPLDCKR19}. Thus, this could also be exploited by the adversary.
If we change the intermediate values, the model output will change, potentially revealing the information about the model parameters.
We focus on utilizing this behavior to fully recover the values of the internal parameters of the neural network.
More specifically, we utilize a fault that changes the sign of the intermediate values to get the information, hence the name \textit{SNIFF -- \underline{s}ig\underline{n} b\underline{i}t \underline{f}lip \underline{f}ault.}

% According to~\cite{breier2018practical}, each activation function behaves in a slightly different way when a fault is applied to it during the computation.
% The assumption here is that if the attacker can observe the change in the behavior, she can gain the knowledge of what function is used in each layer of the network.

\textbf{Our contribution.} In this paper, we present a way to reverse engineer neural networks with the help of fault injection attacks.
More specifically, we target deep-layer feature extractor networks produced by transfer learning, to recover the parameters (weights and biases) of the last hidden layer.
Our work mainly focuses on neural network classifiers based on the widely used softmax activation function in the output layer.
On top of that, we also show application to other activation functions.
Our method provably allows \textit{exact extraction}, meaning that the exact values of parameters can be determined after the fault attack.
Thus, in case of a deep-layer feature extractor, this allows to get the exact information on the entire network.
We note that this is the first work using fault injection attack for the model extraction, and also the first work allowing exact extraction.
In terms of \textit{accuracy} and \textit{fidelity} defined recently in~\cite{jagielski2019high}, our work achieves perfect scores in both, as the extracted values are identical to the original network.

\section{Preliminaries}
\label{sec:related}

This section recalls general concepts used in the rest of the paper.
The target datasets and experimental setup are also discussed.

\subsection{Fault Injection Methods}
\label{sec:fi_methods}
Fault injection can be performed with a variety of equipment based on the required precision, cost and impact.

\textit{Clock/voltage glitch} 
%can be achieved using inexpensive equipment that either varies the external clock signal to the device or under-powers the supply voltage to the chip. 
methods offer limited precision and are normally used to alter the control flow of the program rather than disturbing the data directly.
This is often referred to as \textit{global} fault injection. % as its effects cannot be localized to a single memory cell.

\textit{Electromagnetic (EM) emanation} is more localized method, where the precision heavily depends on the resolution of the injection probe.
%To disturb digital circuits, the attacker uses a high voltage pulse generator that injects a sudden EM spike through the injection probe.
It was shown that precise bit sets and resets in memory cells can be achieved~\cite{ordas2014evidence}.

\textit{Optical radiation} includes methods with varying precision, using equipment ranging from camera flashes to lasers.
%The main disadvantage of these methods is the need of de-packaging of the device so that the chip components are visible to the light beam.
The advantage is high reproducibility of faults and great precision -- precise bit flips were shown to be possible with lasers.

\textit{Rowhammer}~\cite{seaborn2015exploiting} and \textit{Voltpwn}~\cite{kenjar2020v0ltpwn} are fault injection methods that do not require a dedicated injection device. Such attacks exploit memory and microarchitectural properties for fault injection, and allow bit flips and software controlled faults.
%For Rowhammer, it was shown that by using a repeated access to DRAM cells, there is a certain probability to flip bits in adjacent rows of the memory. This method was used in~\cite{hong2019terminal} to achieve accuracy loss of deep learning models. Voltpwn like attacks exploit dynamic voltage and frequency scaling to inject software controlled faults, which can also affect executions in security enclaves like SGX.

Besides these, there are other less researched fault injection methods, such as X-rays/gamma rays~\cite{anceau2017nanofocused}, or hardware trojans~\cite{breier2015multiple}.
%While these can be very powerful, their practicality is limited either because of strong attacker assumptions or the cost of the injection device.

\subsection{Fault Injection on Neural Networks}
\label{sec:related_fia}
The seminal work in the field of adversarial fault injection was published by Liu et al. in 2017~\cite{liu2017fault}. They introduced two types of attacks: \textit{single bias attack} changes the bias value in either one of the hidden layers (in case of ReLU or similar activation function) or output layer of the network to achieve the misclassification; while \textit{gradient descent attack} works in a similar way as Fast Gradient Sign Method~\cite{goodfellow2014explaining}, but changes the internal parameters instead of the input to the network.

Practical fault injection by using a laser technique was shown by Breier et al. in 2018~\cite{breier2018practical}.
They were able to disturb the instruction execution inside the general-purpose microcontroller to achieve the change of the neuron output.
In their paper, they focused on behavior of three activation functions: in case of sigmoid and tanh, the fault resulted in an inverted output, while in case of ReLU, the output was forced to be always zero.
The work was further extended in~\cite{hou2020security} and~\cite{hou2021physical} to show different attack strategies to improve the misclassification efficiency.

A comprehensive evaluation of bitwise corruptions on various deep learning models was presented by Hong et al. in 2019~\cite{hong2019terminal}.
They showed that most models have at least one parameter such that if there is a bit-flip introduced in its bitwise representation, it will cause an accuracy loss of over 90\%.

Malicious bit-flips were further investigated for various misclassification/model degradation attacks in~\cite{yao2020deephammer,rakin2019bit,rakin2020tbt,bai2021targeted}.

When it comes to fault and error tolerance of neural networks, we would point interested reader to a survey written by Torres-Huitzil and Girau in 2017~\cite{torres2017fault}, which provides exhaustive overview of this topic.

\subsection{Transfer Learning}
\label{sec:transfer}
Transfer learning takes a pre-trained \textit{teacher} model and transfers the knowledge (model architecture and weights) to a \textit{student} model. 
The requirement is to have a similar task for the newly trained student model compared to the teacher model.
Transfer learning is normally achieved by ``freezing'' the first $n-k$ layers of the teacher model out of the total number of $n$ layers -- by fixing the values of the weights.
Then, the remaining $k$ layers are removed and new layers are added to the end of the student model.
These layers are then trained on the new data.
There are 3 main approaches that are used in transfer learning \cite{wang2018great}:
\begin{itemize}
    \item \textit{Deep-layer Feature Extractor:} in this approach, the first $n-1$ layers are frozen and only the last layer is updated, as can be seen in Figure~\ref{fig:transfer}. It is normally used when the student task is very similar to the teacher task. It allows very efficient training.
    In the rest of the paper, we will be focusing on the \textit{secret parameter recovery} of this approach.
    \item \textit{Mid-layer Feature Extractor:} this approach freezes the first $n-k$ layers, where $k < n-1$. It can be used in case the student task is less similar to the teacher task and there is enough data to train the Student.
    \item \textit{Full Model Fine-tuning:} in this approach, all the layers are unfrozen and updated during the student training. It requires sufficient amount of data to fully train the student, and is normally used for cases where student task differs significantly from the teacher task.
\end{itemize}

Important observation when recovering the student model is that the layers copied from the teacher are publicly known, and therefore it is possible to derive the output values for all the frozen layers for any input.
This way, we know the inputs to the layers trained by student, and the outputs from the model.
Based on this information, we are able to design a weight recovery attack assisted by fault injection.

\begin{figure}[tb]
    \centering
    \includegraphics[width=0.48\textwidth]{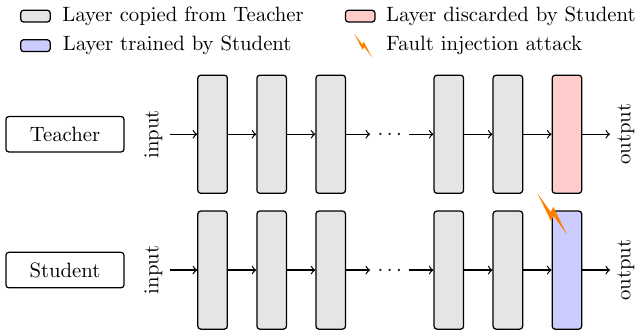}
    \caption{Transfer learning using deep-layer feature extractor and fault injection into the student model for recovering the newly added layer.}
    \label{fig:transfer}
\end{figure}

\subsection{Model Extraction}
\label{sec:related_extraction}
If we consider $\mathcal{O}(\cdot)$ to be the original neural network model we want to extract, $\hat{\mathcal{O}}(\cdot)$ denotes the extracted model.
Jagielski et al.~\cite{jagielski2019high} developed a taxonomy for model extraction attacks and differentiate four different extraction types:
\begin{itemize}
    \item \textit{Exact Extraction:} strongest type of extraction, where $\hat{\mathcal{O}} = \mathcal{O}$, that is, both the architecture and the weights of the extracted model have the same values as the original network.  Exact extraction brings several advantages to the adversary over other extraction types. Firstly, it aides in computing perfect adversarial examples~\cite{shamir2019simple}, which is considered one of the most powerful and stealthy exploits against neural networks. Secondly, the knowledge of exact model also reveals information on training data~\cite{shokri2017membership} which can be highly sensitive and proprietary. It was shown to be \textbf{impossible} to do such extraction for many types of neural networks in black-box \textit{fault-free} scenario, and therefore \cite{jagielski2019high} only focuses on the following three attacks.
    \item \textit{Functionally Equivalent Extraction:} slightly weaker assumption is considered for functionally equivalent extraction, where the attacker is capable of constructing $\hat{\mathcal{O}}$ such that $\forall x \in \mathcal{X}, \hat{\mathcal{O}}(x) = \mathcal{O}(x)$. 
    In such case, it is not necessary to match the two models exactly, only the output of both models has to be the same for all the elements from the domain $\mathcal{X}$ of the dataset $\mathcal{D}$.
    \item \textit{Fidelity Extraction:} for a target distribution $\mathcal{D}_F$ over $\mathcal{X}$, and goal similarity function $S(p_1, p_2)$, fidelity extraction tries to construct $\hat{\mathcal{O}}$ that minimizes Pr$_{x\sim\mathcal{D}_F} \left[S(\hat{\mathcal{O}}(x),\mathcal{O}(x)) \right]$.
    The adversary normally wants to keep both the correct and incorrect classification between the two models.
    A functionally equivalent extraction achieves a fidelity of 1 on all distributions and all distance functions.
    \item \textit{Task Accuracy Extraction:} for a true task distribution $\mathcal{D}_A$ over $\mathcal{X} \times \mathcal{Y}$, task accuracy extraction tries to construct an $\hat{\mathcal{O}}$ that maximizes Pr$_{(x,y)\sim\mathcal{D}_A} \left[\text{arg max}(\hat{\mathcal{O}}(x))=y\right]$. In this setting, the aim is to achieve the same or higher accuracy than the original model. Therefore, it is the easiest type of extraction attack to construct, as it does not care about the original model's mistakes.
\end{itemize}

\section{Methods}
\label{sec:methods}
To be able to reverse engineer a neural network with fault injection attack, we first need to know the erroneous behavior of its elementary components -- neurons.
To study this behavior, we first identify each part of a neuron that can be faulted.

\subsection{Possibilities to Fault a Neuron}
\label{sec:neuron}
Figure~\ref{fig:activation} shows a typical neuron computation in a neural network.
Inputs are multiplied with weights and then summed together, adding a bias.
Resulting value is fed to the activation function, which produces the final output of the neuron.
Below, we identify the points where a fault can be introduced (numbers correspond to those in Figure~\ref{fig:activation}):
\begin{enumerate}
    \item[1.] \textit{Inputs:} there are two possibilities to fault the input -- either at the output of a neuron from the previous layer or at the input of the multiplication of the current neuron. The first case affects the computation of all the neurons in the current layer, while the second case only affects the target neuron.
    \item[2.-3.] \textit{Weights, Product:} unlike faulting the input, weight or product change only affects the target neuron. As we explain later in this paper, attacks on these values can give the attacker knowledge of the weights.
    \item[4.-5.] \textit{Bias, Summation:} attacks on bias can slightly change the input to the activation function, while the attacks on summation can change this greatly. Therefore, the latter one can be considered as one of the means of misclassification by faults.
    \item[6.] \textit{Activation function:} Fault attacks on activation function were studied in~\cite{breier2018practical} from instruction skipping perspective. If attacked with a sufficient precision, they can cause misclassification.
\end{enumerate}

\begin{figure}[tb]
    \centering
    \includegraphics[width=0.49\textwidth]{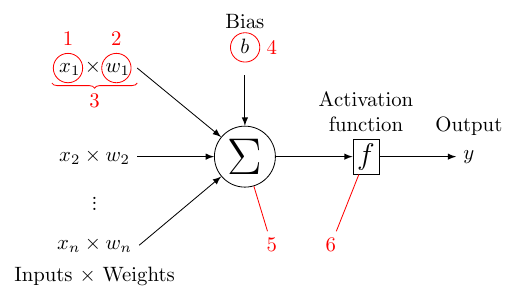}
    \caption{Neuron computation of neural networks.}
    \label{fig:activation}
\end{figure}

\subsection{Experimental Setup}
\label{sec:setup}
In this work, we consider different models which were pretrained using transfer learning~\cite{wang2018great} on ImageNet dataset, following deep-layer feature extractor approach. 
We use several models which are available in public libraries, such as Keras~\cite{chollet2015keras} and PyTorch~\cite{pytorch}, and for the experiments, the last fully connected layers are removed and substituted with single fully connected layer, and retrained. 
For the training data, the visual dataset for object recognition task, CIFAR-10 \cite{Krizhevsky09learningmultiple}, is used. 
The CIFAR-10 dataset contains 50k training data, and 10k test data, each of which is a $32\times 32$ pixels color image. First, the images are upscaled to be consistent with the dimension used in the pretrained model, followed by normalization. 
Next, we add a Dense layer with 10 neurons at the output, corresponding to 10 classes in the dataset.
The activation function used for the output layer is softmax.
%This is illustrated in Figure~\ref{fig:last_layers}. 
Global Average Pooling or Flatten is used before the dense layer to reduce the number of neurons at the output of pretrained networks. 

\subsection{Adversary Model}
We consider an adversary model, where the adversary aims at IP theft for overproduction and illegal cloning of ML proprietary models, running on edge/IoT devices. The proprietary ML models are carefully derived through transfer learning from popular and open ML models like AlexNet~\cite{krizhevsky2012imagenet}, VGG-16~\cite{simonyan2014very}, ResNet-50~\cite{DBLP:conf/cvpr/HeZRS16}, Inception V3~\cite{DBLP:conf/cvpr/SzegedyVISW16}, etc. While the initial layers are publicly known, the adversary aims at recovering the parameters of the re-trained fully connected layers.
To enable model recovery, adversary acquires few legal copies of the target. Being a legal user, the adversary can use the target devices with known data and inject faults into the device. Fault injection is followed by secret parameters recovery. This is a case of IP theft that allows adversary to overproduce/clone the ML model on huge number of devices without paying the legal licence fee.

% \subsection{Fault Injection Testing Framework}
% \label{sec:framework}
% \jakub{This part will be removed}
% Our framework for testing neuron computation under faults is depicted in Figure~\ref{fig:framework}.
% It consists of following components:
% \begin{itemize}
%     \item \textit{Input generator:} generates inputs to the neuron according to evaluation strategy. As it is not possible to test the entire input space, this strategy has to be defined in a way that it covers the most likely input values.
%     \item \textit{Neuron computation:} takes the input generated by input generator and does the computation according to neuron parameters. These parameters are number of inputs, type of activation function, weight values, and bias.
%     \item \textit{Fault injection:} injects a fault during the neuron computation according to fault injection parameters. Type of fault determines the way how the computations will be disturbed -- skip of operations, random change of value, bit flip of value, etc. Location specifies which part of the neuron computation will be affected.
%     \item \textit{Output checker:} checks whether the output value has changed, and if yes, how far is the faulty value from the expected value.
% \end{itemize}

% \begin{figure}[tb]
%     \centering
%     \includegraphics[width=0.45\textwidth]{fig/neuron_injection.pdf}
%     \caption{Fault injection testing framework that is capable of evaluating neuron computation under fault injection attack.}
%     \label{fig:framework}
% \end{figure}

\subsection{SNIFF -- \underline{S}ig\underline{n} B\underline{i}t \underline{F}lip \underline{F}ault}
\label{sec:sniff}
The attack model for our work is bit flip on the sign bit of the intermediate values.
%We named this type of attack \textit{SNIFF -- \underline{s}ig\underline{n} b\underline{i}t \underline{f}lip \underline{f}ault}.
In particular, we consider attack on two intermediate values: SNIFF on the product of the weight and the input, and SNIFF on the bias value.

SNIFF attack on the product can be achieved in the real device by targeting either the input, the weight, or the final product value (targets 1, 2, and 3 in Figure~\ref{fig:activation}).
In Section~\ref{sec:recovery}, we use the bit flip fault on the weight to model this attack.
In case of SNIFF attack on the bias value, the attacker has to target the bias itself (target 4 in Figure~\ref{fig:activation}).

\subsection{Finding the Correct Timing for Faults}
\label{sec:sca}

Once the target step is identified, one needs to find precise timing locations corresponding to the sensitive computation. As already demonstrated in~\cite{batina2018csi}, it is possible to determine the timing by using side-channel information, coming either from the power consumption of the device or from electromagnetic emanation (EM).

It can be shown in the example of 4 fully connected layer with 50, 30, 20 and 50 neurons in each layer respectively from the input layer, on ARM Cortex-M3 microcontroller mounted on the Arduino Due. The electromagnetic emanation measured through a near field probe 
(RF-U 5-2 H-field probe from Langer) 
is shown in Fig.~\ref{fig:sca}. In Fig~\ref{fig:sca} (a) each layer can be easily identified. Next, Fig~\ref{fig:sca} (b) shows a zoom on computation of the first neuron of the third layer. Given the (50, 30, 20, 50) architecture, 20 multiplications are expected followed by the activation function. Each multiplication can be easily identified in Fig~\ref{fig:sca} (b) and thus precisely targeted with faults.

The process of finding the correct timing can be automated by using pattern recognition techniques to locate the multiplication patterns within the neuron computations.
Similarly, position on the chip which leaks the information in the form of an electromagnetic field can be automatically located.
For example, \cite{danial2020scniffer} shows both processes.

\begin{figure*}
    \centering
    \begin{tabular}{cc}
       \includegraphics[width=0.49\textwidth]{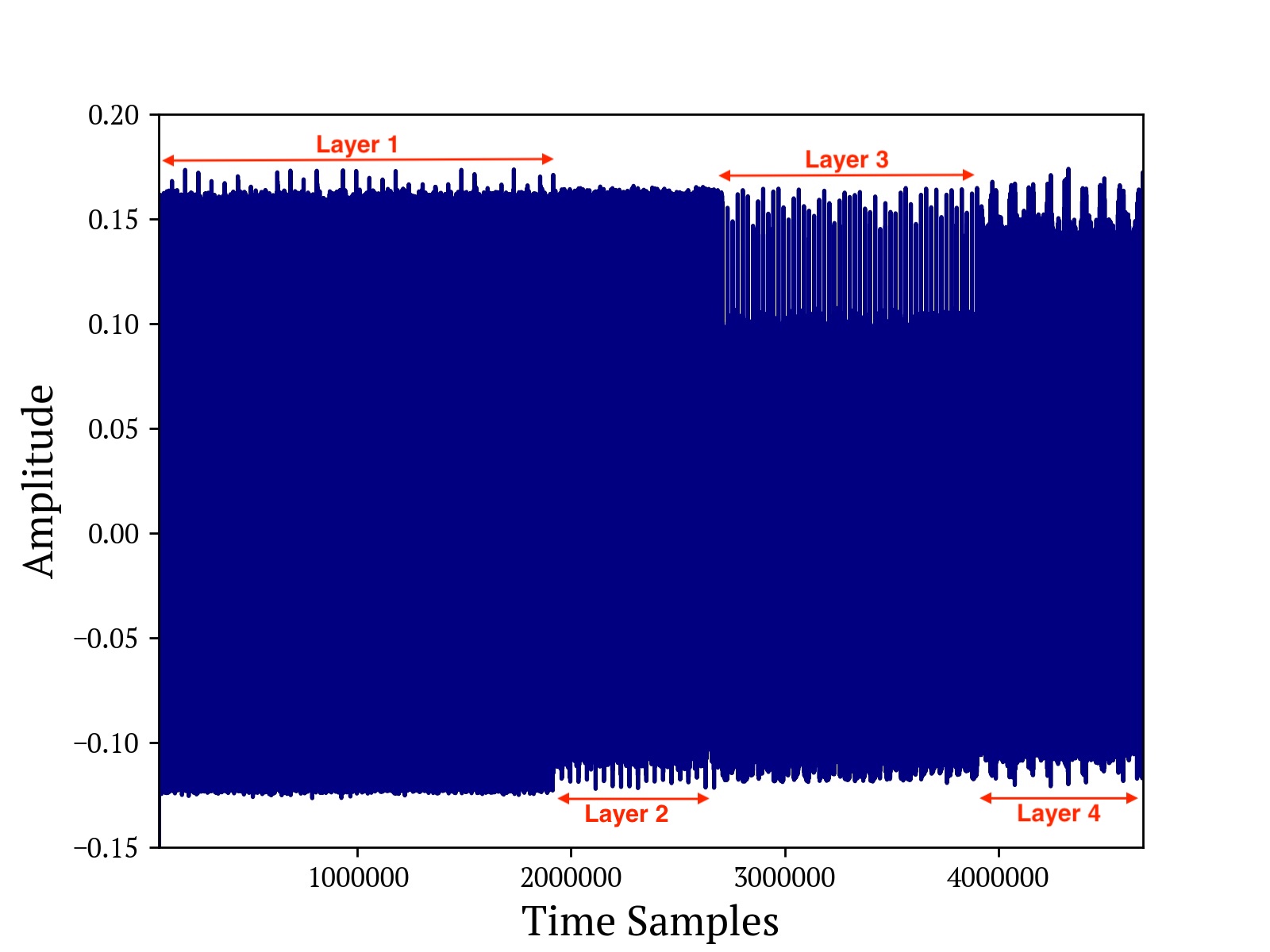}  & \includegraphics[width=0.49\textwidth]{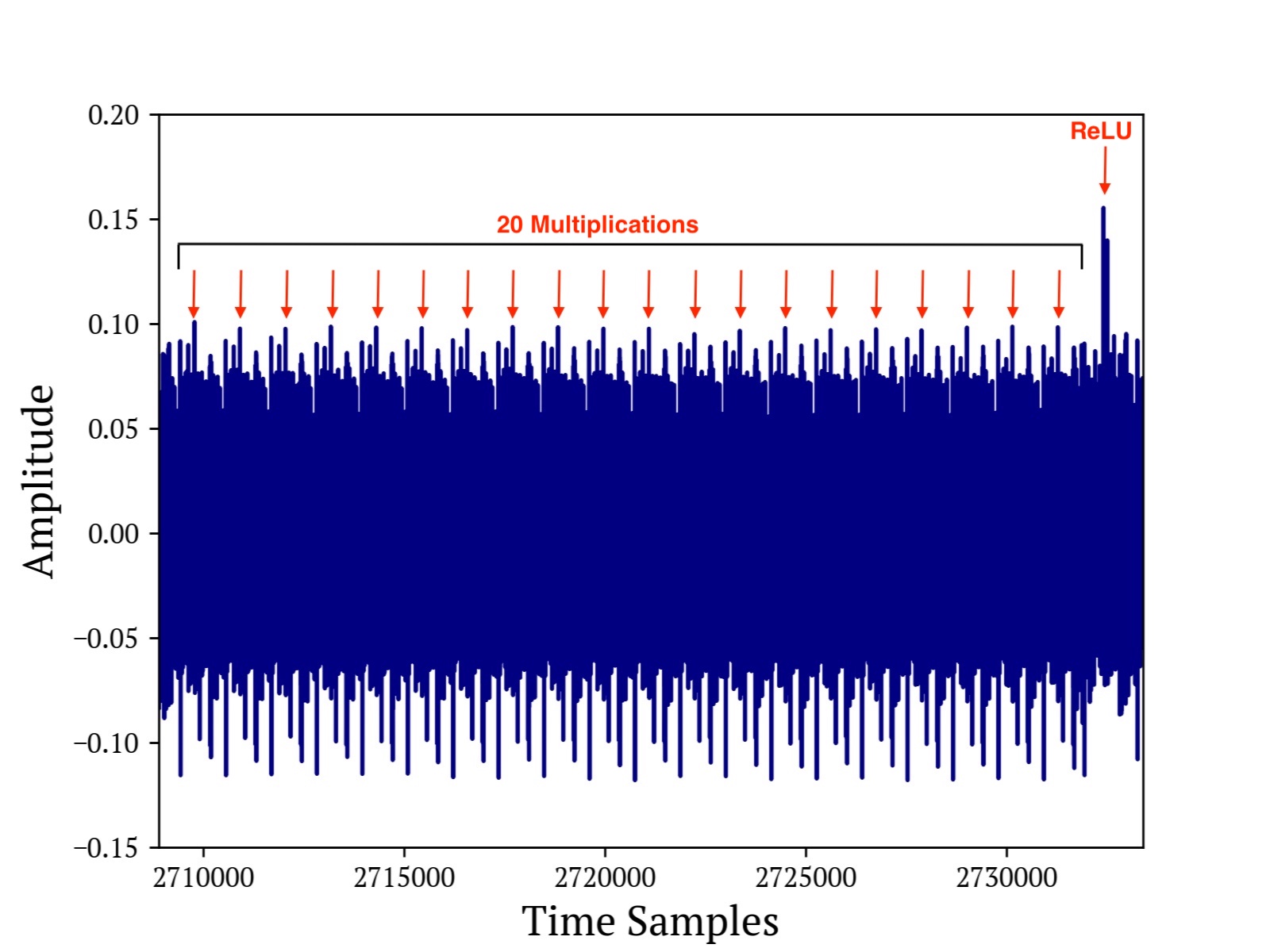} \\
        (a) & (b)\\ 
    \end{tabular}
    \caption{Electromagnetic emanation measurement during the computation of 4 fully connected layers with 50, 30, 20, 50 neurons in each layer. In (a) each layer can be uniquely identified by the measurement trace, while (b) shows execution of one neuron in third layer showing timing of each of the 20 multiplications.}
    \label{fig:sca}
\end{figure*}

%\begin{figure}
%    \centering
%    \includegraphics[width=0.9\linewidth]{sca_detail_ann.jpg}
%    \caption{Electromagnetic emanation measurement during the execution of neural network.}
 %   \label{fig:sca_detail}
%\end{figure}

\section{Recovery of Secret Parameters}
\label{sec:recovery}
In this section, we will explain the recovery of the weights and biases of the last layer of deep-layer feature extractor model, constructed by using transfer learning.

\subsection{Attack Intuition}
The intuition of the parameter recovery attack is as follows. 
As shown in Figure~\ref{fig:transfer}, the attack works on the last layer of the student network.
The detail of this layer is illustrated in Figure~\ref{fig:last_layers}.
The attacker first executes the model computation on last layer input, denoted by $\boldsymbol{I}=(I_1,I_2,\dots,I_n)$, without fault injection, and observes the outputs -- classes and corresponding probabilities from the last softmax layer.

Then, she injects fault into the last layer by performing SNIFF on a single product of the weight and the input ($I_i\times w_{ij}$).
Based on the original (non-faulty) output values and the faulty ones, she can recover the unknown weight $w_{ij}$.
Similarly, by performing SNIFF on a single bias value ($b_{S,i}$), she can recover this value by comparing the faulty and original network outputs.

\begin{figure}[tb]
    \centering
    \includegraphics[width=0.38\textwidth]{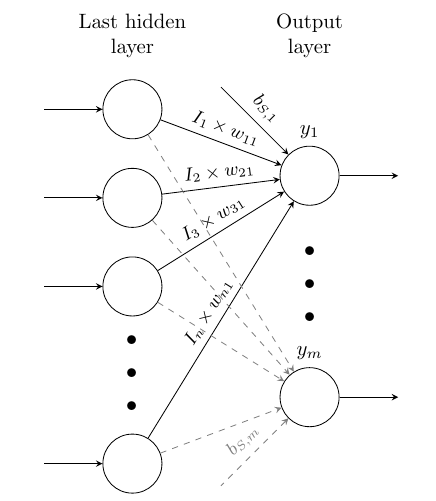}
    \caption{Last two layers of the student model -- nodes $I_i$ are known, while the weights $w_{ij}$ and biases $b_{S,j}$ are the target for the recovery.}
    \label{fig:last_layers}
\end{figure}

\subsection{Formalization}
\label{sec:formalization}
In this section we formally describe the attack.
Suppose there are $k$ layers in the teacher neural net, and for an input $\boldsymbol{x}$, the output is given by $\mathcal{L}_k(\mathcal{L}_{k-1}(\dots \mathcal{L}_1(\boldsymbol{x})))$,
where $\mathcal{L}_i$ denotes the function at layer $i$, which takes the output of the previous layer and gives input for the next layer.
For example, $k=1$ and $\mathcal{L}(\boldsymbol{x})=\text{sigmoid}(\boldsymbol{x}^TW+\boldsymbol{b})$ denotes a fully connected one layer network with weight matrix $W$, bias vector $\boldsymbol{b}$ and activation function sigmoid.

Let $\mathcal{O}_{\theta_{T,-1}}$ denote the part of the teacher neural network that was preserved by the student neural network, i.e.
\[
\mathcal{O}_{\theta_{T,-1}}(\boldsymbol{x}) := \mathcal{L}_{k-1}(\dots \mathcal{L}_1(\boldsymbol{x})).
\]
Here $\theta_{T,-1}$ denotes the parameters of the first $k-1$ layers of the teacher neural network.

Let $W_S$ and $\boldsymbol{b}_S$ denote the trained weight matrix and bias vector for the last layer of student neural network.
Suppose the $(k-1)$th layer of teacher network has $n$ neurons and the output layer of student network has $m$ neurons.
Then we have $W_S$ is an $n\times m$ matrix and $\boldsymbol{b}_S$ is a vector of length $m$. 
For an input $\boldsymbol{x}$, the output of the student neural network is then given by
\[
\mathcal{O}_{\theta}(\boldsymbol{x})=\text{Softmax}(\mathcal{O}_{\theta_{T,-1}}(\boldsymbol{x})^TW_S+\boldsymbol{b}_S),
\]
Let $\boldsymbol{y}(\boldsymbol{x}):=\mathcal{O}_{\theta_{T,-1}}(\boldsymbol{x})^TW_S+\boldsymbol{b}_S$, then we have for $i=1,2,\dots,m$,
\[
\mathcal{O}_{\theta,i}(\boldsymbol{x}) = \frac{\exp{y_i(\boldsymbol{x})}}{\sum_{j=1}^m \exp (y_j(\boldsymbol{x}))}.
\]

By our assumption, the attacker knows the teacher neural network and she can also observe the Softmax output, in particular, she knows the number $m$ and hence the dimensions of $W_S$ and $\boldsymbol{b}_S$.
% knows the architecture of the student neural network.
Our goal of model extraction then consists of recovering $\theta$, the parameters for the student neural network.
Let $\theta_S:=\{W_S,\boldsymbol{b}_S\}$, then $\theta = \theta_S\cup\theta_{T,-1}$.
Note that $\theta_{T,-1}$ are the parameters from the teacher network, which are public information.
Thus our goal is to recover $\theta_S$, or equivalently, $W_S$ and $\boldsymbol{b}_S$.

\begin{definition}
An input $\boldsymbol{x}$ is called a \textit{non-vanishing input} for $i$ ($i=1,2,\dots,n$) if $\mathcal{O}_{\theta_{T,-1},i}(\boldsymbol{x})\neq0$.
\end{definition}

%\textcolor{red}{We note that by our assumption, with the knowledge of the teacher neural network, the attacker can calculate $\mathcal{O}_{\theta_{T,-1},i}(\boldsymbol{x})$ for any $\boldsymbol{x}$ without running the student neural network. This part can be considered as an offline phase for our attack}
For simplicity, let $\boldsymbol{I}(\boldsymbol{x})$ denote $\mathcal{O}_{\theta_{T,-1}}(\boldsymbol{x})$.
As described in Section~\ref{sec:sniff}, we consider SNIFF on the product $I_iw_{ij}$ and on the bias $b_{S,j}$.

We refer to the unknown weight $w_{ij}$ as the \textit{target weight parameter} and the unknown bias $b_{S,j}$ as the \textit{target bias parameter}.
\begin{thm}\label{thm:bias}
For any $j_0\in\{1,2,\dots,m\}$ and any input $\boldsymbol{x}$.
Suppose a SNIFF on target bias parameter $b_{S,j_0}$ was carried out.
Let $z_{j_0}$ and $\tilde{z}_{j_0}$ denote the correct and faulted value of $\mathcal{O}_{\theta,{j_0}}(\boldsymbol{x})$. 
Then the target weight $b_{S,j_0}$ can be recovered as:
\[
b_{S,j_0} = \frac{1}{2}\ln\left(\frac{\tilde{z}_{j_0}^{-1}-1}{z_{j_0}^{-1}-1}\right).
\]
\end{thm}
\begin{proof}
Let $j_0$ be given and let $\boldsymbol{x}$ be any input.
For simplicity, we write $\boldsymbol{I}$ (resp. $\boldsymbol{y}$) instead of $\boldsymbol{I}(\boldsymbol{x})$ (resp. $\boldsymbol{y}(\boldsymbol{x})$).
For any $j\in\{1,2,\dots,m\}$,
\[
y_{j} = b_{S,j} + \sum_{i=1}^n I_iw_{ij},\ 
z_{j} = \frac{\exp (y_{j})}{\sum_{j'=1}^m \exp (y_{j'})}
\]
In particular,
\[
y_{j_0} = b_{S,j_0} + \sum_{i=1}^n I_iw_{ij_0},\ 
z_{j_0} = \frac{\exp (y_{j_0})}{\sum_{j=1}^m \exp (y_j)}.
\]
Let
\begin{eqnarray*}
A &:=& \sum_{i=1}^n I_iw_{ij_0} = y_{j_0}-b_{S,j_0},\\
B &:=& \sum_{j=1,j\neq j_0}^m\exp (y_j).
\end{eqnarray*}
We have
\begin{eqnarray*}
z_{j_0}&=&\frac{\exp(b_{S,j_0}+A)}{\exp(b_{S,j_0}+A)+B},\\ \tilde{z}_{j_0}&=&\frac{\exp(-b_{S,j_0}+A)}{\exp(-b_{S,j_0}+A)+B}.
\end{eqnarray*}

We note that by definition of Softmax, $z_{i_0}>0$ and $\tilde{z}_{i_0}>0$.
\begin{eqnarray*}
\frac{1}{z_{j_0}}-1&=&\frac{\exp(b_{S,j_0})\exp(A)+B}{\exp(b_{S,j_0})\exp(A)} - 1\\ &=& \frac{B}{\exp(b_{S,j_0})\exp(A)}
= \exp(-b_{S,j_0})\frac{B}{\exp (A)}.
\end{eqnarray*}
Similarly,
\begin{eqnarray*}
\frac{1}{\tilde{z}_{j_0}}-1&=&\frac{\exp(-b_{S,j_0})\exp(A)+B}{\exp(-b_{S,j_0})\exp(A)} - 1\\
&=& \frac{B}{\exp(-b_{S,j_0})\exp(A)} = \exp(b_{S,j_0})\frac{B}{\exp (A)}.
\end{eqnarray*}
By definition of Softmax, $z_{j_0}^{-1}>1$,
\[
\frac{\tilde{z}_{j_0}^{-1}-1}{z_{j_0}^{-1}-1} =  \exp(2b_{S,j_0})\Longrightarrow b_{S,j_0} = \frac{1}{2}\ln\left(\frac{\tilde{z}_{j_0}^{-1}-1}{z_{j_0}^{-1}-1}\right).
\]
\end{proof}
\begin{cor}\label{cor:bias}
The attacker can recover the bias vector $\boldsymbol{b}_S$ with $m$ faults and $2m$ executions of the target neural network (the student neural network). 
\end{cor}

\begin{thm}
For any $i_0\in\{1,2,\dots,n\}, j_0\in\{1,\dots,m\}$ and any $\boldsymbol{x}$, a non-vanishing input for $i_0$.
Suppose a SNIFF on target weight parameter $w_{i_0j_0}$ was carried out.
Let $z_{j_0}$ and $\tilde{z}_{j_0}$ denote the correct and faulted value of $\mathcal{O}_{\theta,{j_0}}(\boldsymbol{x})$. 
Then the target weight $w_{i_0j_0}$ can be recovered as:
\[
w_{i_0j_0} = \frac{1}{2I_{i_0}}\ln\left(\frac{\tilde{z}_{j_0}^{-1}-1}{z_{j_0}^{-1}-1}\right).
\]
\end{thm}
\begin{proof}
Let $i_0,j_0$ be given, and let $\boldsymbol{x}$ be a non-vanishing input for $i_0$.
For simplicity, we write $\boldsymbol{I}$ (resp. $\boldsymbol{y}$) instead of $\boldsymbol{I}(\boldsymbol{x})$ (resp. $\boldsymbol{y}(\boldsymbol{x})$).
We let $w_{ij}$ denote the $(i,j)$th entry of the weight matrix $W_S$.
And let $b_{S,j}$ denote the $j$th entry of the bias vector $\boldsymbol{b}_S$.
Then for any $j\in\{1,2,\dots,m\}$,
\[
y_{j} = b_{S,j} + \sum_{i=1}^n I_iw_{ij},\ 
z_{j} = \frac{\exp (y_{j})}{\sum_{j'=1}^m \exp (y_{j'})}
\]
In particular,
\[
y_{j_0} = b_{S,j_0} + \sum_{i=1}^n I_iw_{ij_0},\ 
z_{j_0} = \frac{\exp (y_{j_0})}{\sum_{j=1}^m \exp (y_j)}.
\]
Let
\begin{eqnarray*}
A &:=& b_{S,j_0} + \sum_{i=1, i\neq i_0}^n I_iw_{ij} = y_{j_0}-I_{i_0}w_{i_0j_0},\\
B&:=&\sum_{j=1,j\neq j_0}^m\exp (y_j).
\end{eqnarray*}

We have
\begin{eqnarray*}
z_{j_0}&=&\frac{\exp(I_{i_0}w_{i_0j_0}+A)}{\exp(I_{i_0}w_{i_0j_0}+A)+B},\\
\tilde{z}_{j_0}&=&\frac{\exp(-I_{i_0}w_{i_0j_0}+A)}{\exp(-I_{i_0}w_{i_0j_0}+A)+B}.
\end{eqnarray*}
We note that by definition of Softmax, $z_{j_0}>0$ and $\tilde{z}_{j_0}>0$.
{\small
\begin{eqnarray*}
\frac{1}{z_{j_0}}-1&=&\frac{\exp(I_{i_0}w_{i_0j_0})\exp(A)+B}{\exp(I_{i_0}w_{i_0j_0})\exp(A)} - 1\\
&=& \frac{B}{\exp(I_{i_0}w_{i_0j_0})\exp(A)} = \exp(-I_{i_0}w_{i_0j_0})\frac{B}{\exp(A)}.
\end{eqnarray*}}
Similarly,
{\small
\begin{eqnarray*}
\frac{1}{\tilde{z}_{j_0}}-1&=&\frac{\exp(-I_{i_0}w_{i_0j_0})\exp(A)+B}{\exp(-I_{i_0}w_{i_0j_0})\exp(A)} - 1 \\
&=& \frac{B}{\exp(-I_{i_0}w_{i_0j_0})\exp(A)} = \exp(I_{i_0}w_{i_0j_0})\frac{B}{\exp(A)}.
\end{eqnarray*}}
Since $\boldsymbol{x}$ is a non-vanishing input for $i_0$, we have $I_{i_0}\neq 0$.
Also by definition of Softmax, $z_{i_0}^{-1}>1$.
Together with the above equations,
\[
\frac{\tilde{z}_{j_0}^{-1}-1}{z_{j_0}^{-1}-1} =  \exp(2I_{i_0}w_{i_0j_0})\Longrightarrow w_{i_0j_0} = \frac{1}{2I_{i_0}}\ln\left(\frac{\tilde{z}_{j_0}^{-1}-1}{z_{j_0}^{-1}-1}\right).
\]
\end{proof}
Thus the attacker can recover an $i_0,j_0$ entry of the weight matrix $W_S$, by first running an offline phase to find a non-vanishing input $\boldsymbol{x}$ for $i_0$, then with two executions of the student neural network - one without fault and one with fault.
\begin{cor}\label{cor:weight}
The attacker can recover the weight matrix $W_S$ with $mn$ faults and $2mn$ executions of the targeted neural network (the student neural network). 
\end{cor}

In practice, during the inference, Softmax might be omitted to save the computation time.
We remark that in this case, Corollaries~\ref{cor:bias} and \ref{cor:weight} still hold and the computations needed will be even easier.
Keeping notations in Theorem~\ref{thm:bias} and the proof, we have
\[
z_{j_0} = y_{j_0} = b_{S,j_0} + A, \tilde{z}_{j_0} = \tilde{y}_{j_0} = -b_{S,j_0} + A,
\]
then, the target bias can be recovered using
\[
b_{S,j_0} = \frac{1}{2}(z_{j_0}-\tilde{z}_{j_0}).
\]
The target weight can be recovered in a similar manner.

\section{Results and Discussion}
\label{sec:results}
Bit-flip attacks have been shown to be practical on embedded devices~\cite{agoyan2010flip}.
Similar results can be obtained by using a Rowhammer in DRAM memories~\cite{seaborn2015exploiting}.
In this part we first simulate the bit-flip attack in the code and then use the formulas from the previous section to reverse engineer the model parameters.
Then, we compare our results to previous works.
Finally, we discuss selection of the model extraction method based on the attack purpose.

\subsection{Experimental Results}
Experimental results for reverse engineering with bit-flips are stated in Table~\ref{tab:results}.
We targeted deep-layer feature extractor networks that were based on publicly available networks, being able to reverse engineer the weights in the last layer.
When it comes to recovery of weights, the weight precision for all except 3 networks was $10^{-14}$, for the remaining cases it was $10^{-13}$.
In case of bias recovery, the precision was always $10^{-14}$.

We would like to highlight that the method from Section~\ref{sec:methods} allows the recovery of the exact weight value if we have arbitrary precision of floating point numbers.
In practice, this depends on the used library, computer architecture, and settings.
For our experiments we used Python with Keras library (version 2.3.1) for deep learning.
This library uses numpy for floating point number representation, offering different precision ranging from 16 to 64 bits\footnote{Numpy supports up to 128-bit floats, but those are not compatible with Keras.}.
In our setting we set the \texttt{float64} to be the default representation to get the most precise results.

\begin{table*}[tb]
    \centering
    \caption{Experimental results for reverse engineering with faults. We targeted deep-layer feature extractor networks based on publicly available networks for image classification.}
    \label{tab:results}
    \begin{tabular}{cccc}\hline%\cline{3-4}
    \multicolumn{2}{c}{} & \multicolumn{2}{c}{Reverse Engineering} \\\hline
        Model & No. of Features To Recover & Weight Precision & Bias Precision\\\hline
         \hline
        AlexNet~\cite{krizhevsky2012imagenet} & 9216 & $10^{-13}$ & $10^{-14}$ \\\hline
        GoogleNet (Inception V1)~\cite{szegedy2015going} & 1024 & $10^{-14}$  & $10^{-14}$ \\\hline
        VGG-16~\cite{simonyan2014very} & 25088 & $10^{-13}$  & $10^{-14}$ \\\hline
        ResNet-50~\cite{DBLP:conf/cvpr/HeZRS16} & 2048 & $10^{-14}$  &  $10^{-14}$\\\hline
        Inception V3~\cite{DBLP:conf/cvpr/SzegedyVISW16} & 2048 & $10^{-13}$  & $10^{-14}$ \\\hline
        Inception ResNet V2~\cite{DBLP:conf/aaai/SzegedyIVA17} & 1536 & $10^{-14}$  & $10^{-14}$ \\\hline
        Wide-ResNet-50-2~\cite{Zagoruyko2016WRN} & 2048 & $10^{-14}$  &  $10^{-14}$\\\hline
        DenseNet-201~\cite{DBLP:conf/cvpr/HuangLMW17} & 1920 & $10^{-14}$  & $10^{-14}$ \\\hline
        Xception~\cite{DBLP:conf/cvpr/Chollet17} & 2048 &  $10^{-14}$ & $10^{-14}$ \\\hline
        ResNeXt-101 32x8d~\cite{Xie2016} & 2048 &  $10^{-14}$ & $10^{-14}$ \\\hline
        NasNet-A (6 @ 4032)~\cite{DBLP:conf/cvpr/ZophVSL18} & 4032 &  $10^{-14}$ & $10^{-14}$ \\\hline
        %\textcolor{red}{DenseNet-161~\cite{huang2017densely}} & 2208 & 86.6\% & &   \\\hline
    \end{tabular}
\end{table*}

% \section{Discussion}
% \label{sec:discussion}

\subsection{Comparison to Prior Work}

The seminal work of Lowd and Meek~\cite{lowd2005adversarial} enabled full model functionally equivalent extraction for linear models. Further, full model functionally equivalent extraction for a 2-layer non-linear neural network was proposed by Milli et al~\cite{milli2019model} in a theoretical setting. When considering extraction of fully implemented neural networks, only two works have come to light.
Batina et al.~\cite{batina2018csi} relied on side-channel leakage on electromagnetic measurements to extract the functionally equivalent model in a known input setting. They reported an error on recovered weight of $2.5\times 10^{-3}$, and full network recovery. Later, Jagielski et al.~\cite{jagielski2019high} proposed two attacks. One of the two attacks enabled full model functionally equivalent extraction for a 2-layer neural network with a weight error of only $9\times 10^{-7}$, which is current state-of-the-art. This method required access to logit values, which is a stronger assumption compared to outputs of the softmax function used in our approach.
The other method they developed enabled full model extraction preserving task accuracy and fidelity.

Compared to these prior works, the goal of our work is exact extraction.
When experimentally testing our method with Keras and Pytorch, the recovered weight error of our fault assisted approach was at most $10^{-13}$. \emph{It must be noted that the stated error is the precision error of the Python libraries used in our experiments. Otherwise, our proposed method can provably recover the exact weights.}
The comparison is summarized in Table~\ref{tab:related}.

\begin{table*}[tb]
    \centering
    \caption{Comparison With Prior Work targeting direct model extraction. $^*$ denotes that technique has null precision error. In our experiments the error reported was at most $10^{-13}$, which is the precision limitation of the used Python libraries.}
    \label{tab:related}
    \begin{tabular}{ccccc}
    \hline
    Attack & Leakage Source & Weight Error & Target Network & Goal\\\hline\hline
    \cite{lowd2005adversarial} & Labels & N/A & Linear models & Functionally equivalent\\\hline
    \cite{milli2019model} & Gradients/logits & N/A & 2-layer neural network & Functionally equivalent\\\hline
    \cite{batina2018csi} & EM Side-Channel & $2.5\times 10^{-3}$ &  Full network & Functionally equivalent\\\hline  
    \cite{jagielski2019high} &  Probabilities/logits & $9\times 10^{-7}$ & 2-layer neural network & Functionally equivalent\\\hline\hline
    This Work & Faults/Probabilities & $0$ $(10^{-13})^*$ & 2-layer neural network & Exact extraction\\\hline
    \end{tabular}
\end{table*}

\begin{figure}
    \centering
    \includegraphics[width=0.8\linewidth]{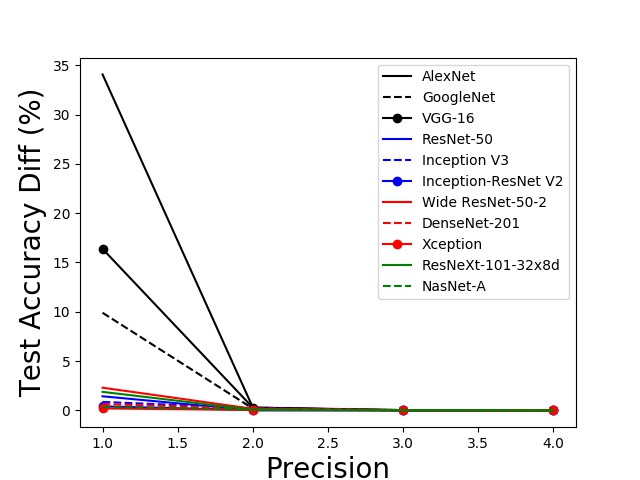}
    \caption{Functionally equivalent model extraction: The difference in test accuracy between the actual model and recovered model against the parameter precision up to certain floating point digit. If the parameter values are the same up to the second decimal point, the test accuracy of the recovered model is the same as the original one for all the evaluated networks.}
    \label{fig:prec_vs_accdiff}
\end{figure}

\subsection{Selecting the Model Extraction Method}
It is important to understand the purpose of the model extraction attack -- after that, it is possible to determine what type of attack should the attacker choose, ultimately deciding the difficulty of the extraction.

If the main goal is to have a task accurate extraction or functionally equivalent extraction, the attacker can achieve this by querying the network with a set of inputs and observing the outputs~\cite{jagielski2019high, milli2019model}.
In this case, the extracted network might have a different architecture than the original one, but will perform well on the same or similar task.
As can be seen in Figure~\ref{fig:prec_vs_accdiff}, for functionally equivalent extraction, it is enough to be able to recover the parameters with the precision of two floating point digits for all the considered networks.
However, if the task changes, the extracted network might give different output than the original one, as it was not trained the same way.
For example, some attackers might be interested in robustness of a certain network to a set adversarial examples, but are not able to query the original network with the entire set.
In such case, task accurate extraction will not help as it will not reveal the vulnerability of the original network by testing the extracted network.
As the adversarial examples are often very close to decision boundaries~\cite{shamir2019simple}, precision of the parameters is crucial to assess the vulnerability.
For such scenarios, it is necessary to have extracted network that is as close to the original network as possible.
That is a task of exact extraction.

% \subsection{Defending against Fault Injection}
% \jakub{TODO}
% There are various techniques that can be employed to provide protection against fault injection attacks.
% We will briefly describe the ones that might be suitable for protecting neural network models.

\section{Protection Techniques}
\label{sec:protection}
In this section we will outline different techniques that can help protect neural network implementations against fault injection attacks.

\subsection{Overview}
In general, the protection techniques against fault injection can work either on device level, or implementation level.

\textit{Device level techniques} focus on preventing the attacker to reach the chip, by various forms of packaging, light sensors, etc.~\cite{bar2006sorcerer}.
The goal is to increase the equipment and expertise requirement to access the chip in a way that the possible reward for the attacker for doing so will be lower than the effort she has to put in.
Device level techniques can also have a different working principle -- to detect potential tampering with the chip.
In this case, a hardware sensor that checks environmental conditions can be deployed~\cite{he2016ring,zussa2014efficiency,ravi2018ppap}.

\textit{Implementation level techniques} aim at detecting changes in the intermediate data.
Detection can be achieved by using various encoding techniques, ranging from simple ones such as parity~\cite{karri2003parity}, to sophisticated codes that can be customized to protect against specific fault models~\cite{breier2019evaluating}.
Another approach is performing the computation several times and comparing the result.
A different way to use redundancy is to perform it at the instruction level, either by generating instruction sequences that replace the original vulnerable instructions~\cite{patranabis2017fault}, or by re-arranging the data within the instructions to make it hard to tamper with without detection~\cite{patrick2016lightweight}.
However, there is no straightforward way of using these two techniques for protecting DNNs.
It is important to mention that unlike device level techniques, the implementation level countermeasures normally incur significant overheads, either in time, circuit area, or power consumption.

\textit{Protecting the learning phase.} Additionally, there is a line of work that focuses on protecting the learning phase of the deep learning method~\cite{taniguchi1999activation}. Such protection technique might be useful in case the learning does not happen in a protected environment and there is a significant risk of faults coming either from the environment or from the attacker.
In our work we consider the model is already learned and therefore, the attacker is trying to tamper with the classification phase.

\subsection{Analysis}
\label{sec:protection_analysis}
Analysis of overheads and coverage of each countermeasure that can be used against instruction skips presented in earlier sections is stated in Table~\ref{tab:countermeasures}.
Here, we provide more details on each technique and its applicability to DNN.

\noindent
\textbf{Spatial/temporal redundancy.} This is the most straightforward way to protect a circuit. Implementer can choose the number of redundant executions depending on what attacker model is expected. In case of redundancy, there is always an integrity check or a majority voting that decides whether the output is valid or not.
When used as a countermeasure in cryptography, circuit is either deployed 2-3$\times$ on the chip (spatial redundancy), or the computation is repeated 2-3$\times$ one after another (temporal redundancy)~\cite{barenghi2010countermeasures}.
Execution times can be randomized so that it is hard to reproduce the same fault in all the redundant executions.

\begin{table*}[tb]
\begin{center}
\caption{Overview of countermeasures effective against skipping instructions.}
\label{tab:countermeasures}
%\footnotesize
\begin{tabular}{cccc}\hline
    \multicolumn{1}{c}{} & \multicolumn{2}{c}{Overhead} & \multicolumn{1}{c}{} \\ \hline
    Countermeasure & Time & Area & Coverage \\
    \hline\hline
    Spatial redundancy ($\times N$) & -- & $N\times 100\%$ & \multicolumn{1}{m{9cm}}{Covers up to $N-1$ faults. To break the countermeasure, faults need to be injected at the same instruction in all the redundant circuits -- which normally requires multiple fault injection devices.}\\ \hline
    Temporal redundancy ($\times N$) & $N\times 100\%$ & -- & \multicolumn{1}{m{9cm}}{Covers up to $N-1$ faults. To break the countermeasure, faults need to be injected at the same instruction in all the redundant executions.}\\ \hline
%    Instruction redundancy~\cite{patrick2016lightweight} & 125-317\% & 14\% & \multicolumn{1}{m{12cm}|}{In its current form~\cite{patrick2016lightweight}, the instruction redundancy cannot be directly applied to neural networks for protecting against instruction skips. It can only be applied to prevent faults on values, where it provides similar fault coverage to a parity protection.}\\ \hline
    Software encoding~\cite{breier2019evaluating} & 75\% & $\approx 65,000\%$ & \multicolumn{1}{m{9cm}}{Protects against instruction skips that target one instruction at a time. Although it does not protect against consecutive instruction skips, during one execution it can protect arbitrary number of non-consecutive skips with 100\% detection rate.} \\ \hline
%    Hardware encoding & & & \\ \hline
    Hardware sensor~\cite{khairallah2019differential} & -- & 1.1\%\footnotemark & \multicolumn{1}{m{9cm}}{As the sensor is based on detecting voltage variations on the chip surface, the detection rate depends on the fault injection device parameters. The most recent work shows high detection rates for both laser and EM fault injection techniques, 97\% and 100\% detected injections, respectively.}\\ \hline
\end{tabular}
\end{center}
\end{table*}

\noindent
\textbf{Software encoding.} As the software encoding countermeasures are realized by table look-up operations, they are not directly applicable to neural networks which operate on real values. However, it is possible to apply this countermeasure for fixed-point arithmetic networks~\cite{hwang2014fixed}.
As it was shown, fixed-point arithmetic can provide good results when used on bigger networks~\cite{sung2015resiliency}. 
The timing overhead in this case is around 75\% -- for example, let us consider a multiplication operation on AVR architecture: for the unprotected implementation, there is operand loading into the registers ($2\times 1$ clk cycle), followed by a multiplication ($2$ clk cycles), resulting into $4$ clock cycles. 
For the protected implementation, there is a register precharge (see e.g. Section 5.1 of~\cite{breier2019evaluating}) of both input registers and the output register ($3\times 1$ clk cycle), followed by the operand loading ($2\times 1$ clk cycle) and table look-up ($2$ clk cycles), resulting into $7$ clock cycles.
Regarding the area overhead, as stated in~\cite{breier2019evaluating}, in case the codeword size is $\leq 8$ bits, there is a fixed table size of $65$ kB per binary operation (e.g. multiplication).
That is why the area (memory) overhead is huge for this case.

\noindent
\textbf{Hardware sensor.} Application of a hardware sensor to protect DNN circuit is depicted in Figure~\ref{fig:sensor}. The main advantage of hardware sensor is that there is no need to change the underlying implementation of the neural network.
The sensor resides on the front side of the chip, protecting all the underlying circuits from fault injection.
In case there is a sudden parasitic voltage detected by such sensor, it raises an alarm.
While front side deployment might be vulnerable to back (substrate) side injection, \cite{he2016ring} reported successful detection of backside injection. Recently, circuit level techniques were also proposed to enhance backside detection capabilities~\cite{matsuda2018286}.
Afterwards, security measures, such as discarding the output, can be applied.
Recently, a way to automate the deployment of such circuit was proposed~\cite{automated_book}.

\begin{figure}[tb]
    \centering
    \includegraphics[width=0.49\textwidth]{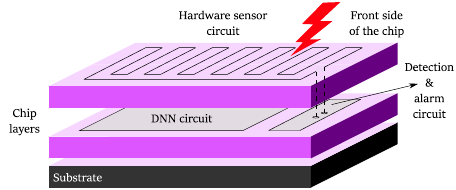}
    \caption{Hardware sensor protecting the DNN circuit.}
    \label{fig:sensor}
\end{figure}

%\vspace{0.1cm}
To summarize, selection of countermeasures depends heavily on the type of application that relies on DNN outputs.
For security critical application, it would be recommended to combine several techniques together to minimize the possible attack vectors and make cost of the attack as high as possible.

\section{Conclusion}
\label{sec:conclusion}
In this paper, we developed a method for provable exact extraction of neural network parameters with the help of fault injection.
Our method aims at recovering the student layer of deep-layer feature extractor networks that were constructed by transfer learning.
This is done by changing the sign of intermediate values to obtain the information about the parameters with a method called SNIFF -- \underline{s}ig\underline{n} b\underline{i}t \underline{f}lip \underline{f}ault.
Our practical experiments show that the exact recovery ultimately depends on computer architecture and the precision of the library used. 
For 64-bit floats used in Keras, the parameter recovery error was at most $10^{-13}$.

For the future work, it would be interesting to look at methods that would allow extraction of parameters from deeper layers of a network.
It would be also worth exploring whether a combination of multiple faults during a single execution can improve the efficiency of the attack.

\bibliographystyle{IEEEtran}
\bibliography{references}

\footnotetext{Sensor requires power during the operation, therefore there is a power overhead of $\approx 5.3\%$ per 16-bit multiplier.}

\begin{IEEEbiography}[{\includegraphics[width=1in,height=1.1in,clip,keepaspectratio]{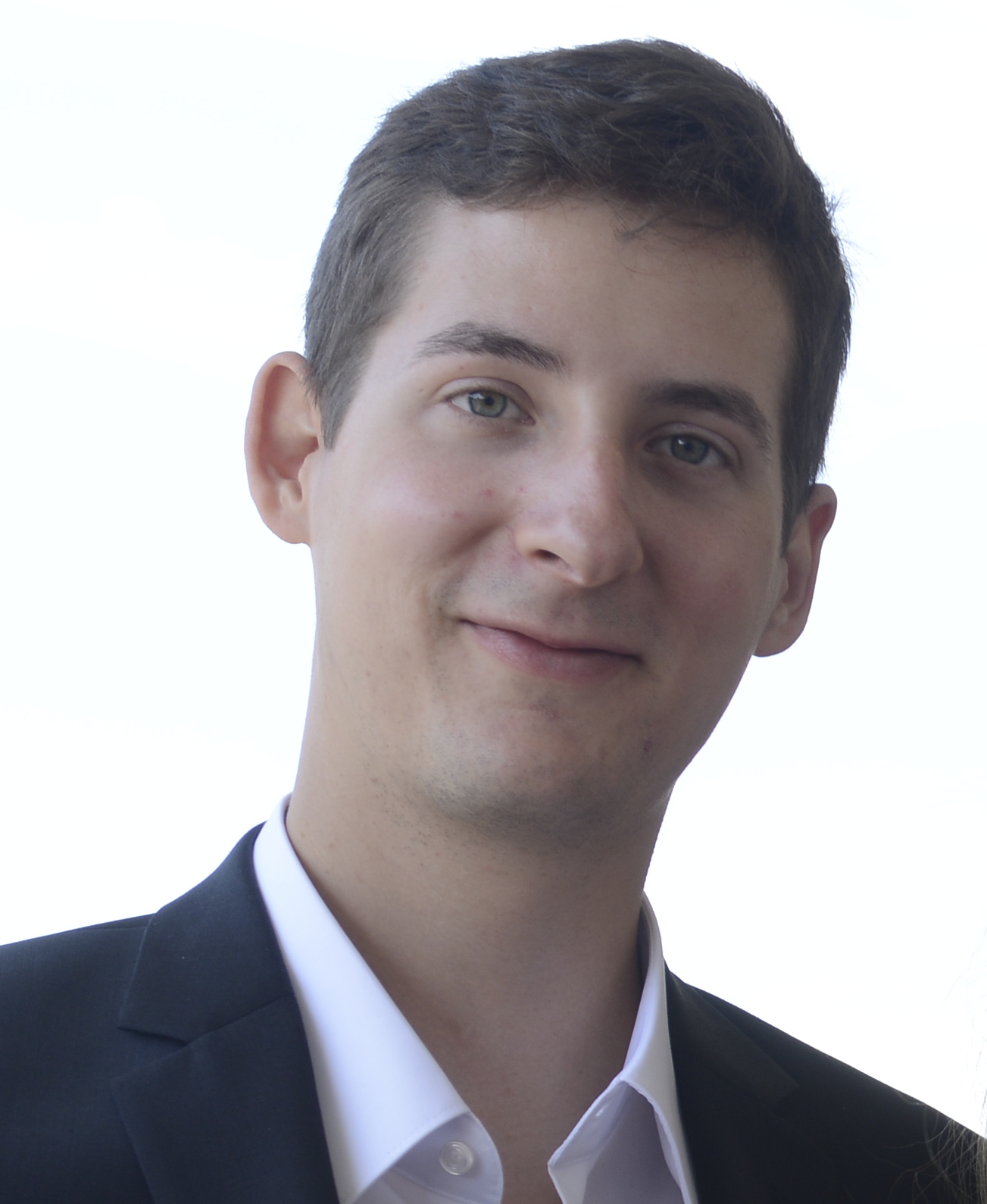}}]%
{Jakub Breier}
is currently a Senior Researcher in Embedded Security at Silicon Austria Labs, Graz, Austria.
Before that, he worked at Nanyang Technological University, Singapore on hardware security and at Underwriters Laboratories, Singapore on security evaluation of embedded devices.
He received his PhD in Applied Informatics from Slovak University of Technology (STU), Slovakia in 2013, Master's in Information Technology Security from Masaryk University, Czech Republic in 2010, and Bachelor's in Informatics from STU, Slovakia in 2008. 
His research topics include fault and side-channel analysis methods and countermeasures, advanced fault injection techniques, and deep learning security.
 \end{IEEEbiography}

%\vspace{-1.7cm}

\begin{IEEEbiography}[{\includegraphics[width=1in,height=1.1in,clip,keepaspectratio]{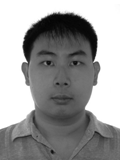}}]%
{Dirmanto Jap}
is currently a Research Scientist at PACE Lab, Temasek Laboratories, Nanyang Technological University (NTU), Singapore. He previously received his Ph.D in Mathematics from NTU in 2016. His main research topics include physical attacks (side-channel and fault attacks) and countermeasures, practical laser/EM fault injection, application of machine learning and deep learning for side-channel attacks and hardware Trojan detection, as well as security of deep learning.
\end{IEEEbiography}

%\vspace{-1.7cm}

\begin{IEEEbiography}[{\includegraphics[width=1in,height=1.1in,clip,keepaspectratio]{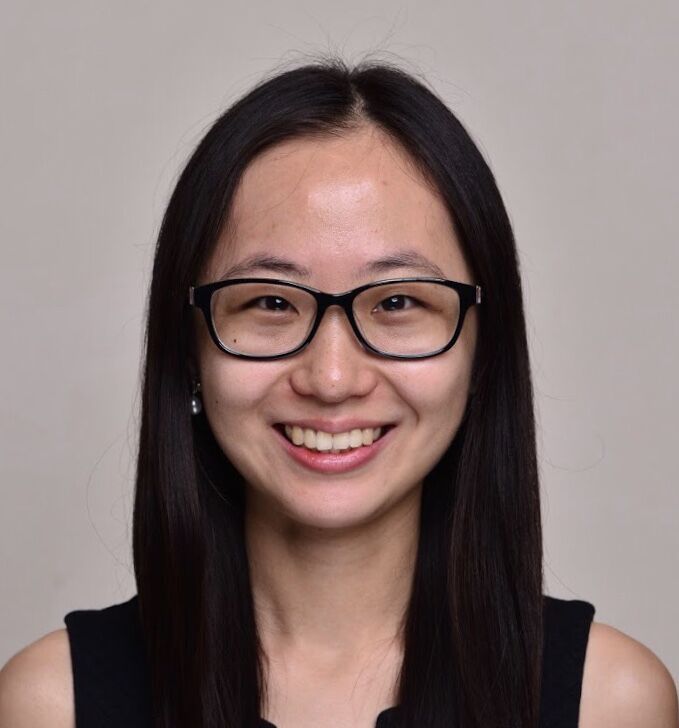}}]%
{Xiaolu Hou} is currently an Assistant Professor at Faculty of Informatics and Information Technologies, Slovak University of Technology, Slovakia.
She received her Ph.D. degree in mathematics from Nanyang Technological University (NTU), Singapore, in 2017.
Her research focus is on fault injection and side-channel attacks. She also has research experience in security of neural networks, location privacy, multiparty computation and differential privacy.
With a wide range of research interests, she has published her work at top venues within various fields, ranging from mathematics to computer security.
\end{IEEEbiography}

%\vspace{-1.7cm}

\begin{IEEEbiography}[{\includegraphics[width=1in,height=1.1in,clip,keepaspectratio]{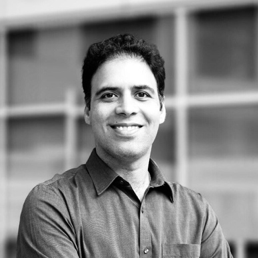}}]%
{Shivam Bhasin} is a Senior Research Scientist and Programme Manager (Cryptographic Engineering) at Centre for Hardware Assurance, Temasek Laboratories, Nanyang Technological University Singapore. He received his PhD in Electronics \& Communication from Telecom Paristech in 2011, Advanced Master (Mastère Spécialisé) in security of integrated systems \& applications from Mines Saint-Etienne, France in 2008. Before NTU, Shivam held position of Research Engineer in Institut Mines-Telecom, France. He was also a visiting researcher at UCL, Belgium (2011) and Kobe University (2013). His research interests include embedded security, trusted computing and secure designs. He has co-authored several publications at recognized journals and conferences. Some of his research now also forms a part of ISO/IEC 17825 standard.
\end{IEEEbiography}

% \vspace*{-2.7cm}

\begin{IEEEbiography}[{\includegraphics[width=1in,height=1.1in,clip,keepaspectratio]{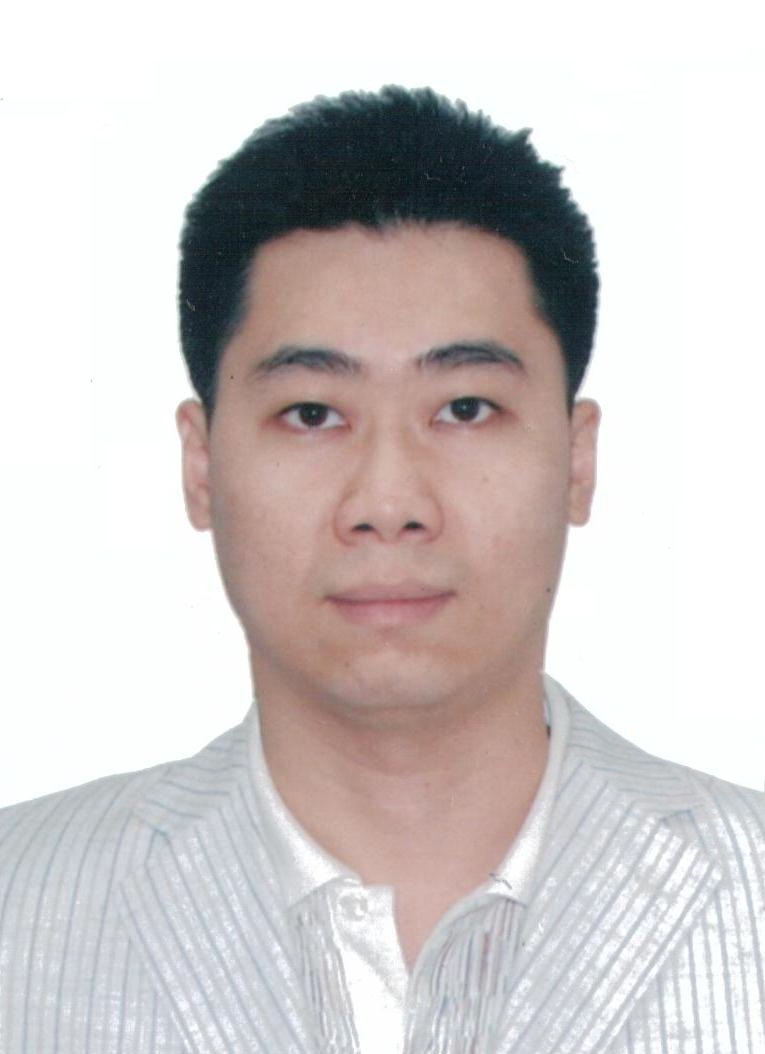}}]%
	{Yang Liu} graduated in 2005 with a Bachelor of Computing (Honours) in the National University of Singapore (NUS). In 2010, he obtained his PhD and started his post doctoral work in NUS, MIT and SUTD. In 2012 fall, he joined Nanyang Technological University (NTU) as a Nanyang Assistant Professor. He is currently an associate professor and Director of the cybersecurity lab in NTU. Dr. Liu specializes in software verification, security and software engineering. His research has bridged the gap between the theory and practical usage of formal methods and program analysis to evaluate the design and implementation of software for high assurance and security. By now, he has more than 200 publications in top tier conferences and journals. He has received a number of prestigious awards including MSRA Fellowship, TRF Fellowship, Nanyang Assistant Professor, Tan Chin Tuan Fellowship, and 8 best paper awards in top conferences like ASE, FSE and ICSE. He is leading a large research team working on the state-of-the-art software engineering and cybersecurity problems.
\end{IEEEbiography}
% \end{small}
\appendix
\subsection{Other Activation Functions}
In this section we consider the case when the activation function of the output layer is not softmax.
Following notations from Section~\ref{sec:formalization}, suppose there are $k$ layers in the teacher neural net, and let $\mathcal{O}_{\theta_{T,-1}}$ denote the part of the teacher neural network that was preserved by the student neural network.
%, i.e.
% \[
% \mathcal{O}_{\theta_{T,-1}}(\boldsymbol{x}) := \mathcal{L}_{k-1}(\dots \mathcal{L}_1(\boldsymbol{x})).
% \]
% Here $\theta_{T,-1}$ denotes the parameters of the first $k-1$ layers of the teacher neural network.
Let $W_S$ and $\boldsymbol{b}_S$ denote the trained weight matrix and bias vector for the last layer of student neural network.
Suppose the $(k-1)$th layer of teacher network has $n$ neurons and the output layer of student network has $m$ neurons.

% Let $\theta_S:=\{W_S,\boldsymbol{b}_S\}$, then $\theta = \theta_S\cup\theta_{T,-1}$.
% Note that $\theta_{T,-1}$ are the parameters from the teacher network, which are public information.
Our attack goal is to recover $W_S$ and $\boldsymbol{b}_S$.
We refer to the unknown weight $w_{ij}$ as the \textit{target weight parameter} and the unknown bias $b_{S,j}$ as the \textit{target bias parameter}.

For an input $\boldsymbol{x}$, let $\boldsymbol{I}(\boldsymbol{x})$ denote $\mathcal{O}_{\theta_{T,-1}}(\boldsymbol{x})$.
Let $\boldsymbol{y}(\boldsymbol{x}):=\mathcal{O}_{\theta_{T,-1}}(\boldsymbol{x})^TW_S+\boldsymbol{b}_S$.
For simplicity, we write $\boldsymbol{I}$ (resp. $\boldsymbol{y}$) instead of $\boldsymbol{I}(\boldsymbol{x})$ (resp. $\boldsymbol{y}(\boldsymbol{x})$).
As described in Section~\ref{sec:sniff}, we consider SNIFF on the product $I_iw_{ij}$ and on the bias $b_{S,j}$.

\subsubsection{Sigmoid}

In case the activation function for the last layer is sigmoid, for an input $\boldsymbol{x}$, the output of the student neural network is given by
\[
\mathcal{O}_{\theta}(\boldsymbol{x})=\text{sigmoid}(\mathcal{O}_{\theta_{T,-1}}(\boldsymbol{x})^TW_S+\boldsymbol{b}_S),
\]
for $i=1,2,\dots,m$,
\[
\mathcal{O}_{\theta,i}(\boldsymbol{x}) = \frac{1}{1+\exp (-y_i(\boldsymbol{x}))}.
\]

\begin{thm}
For any $j_0\in\{1,2,\dots,m\}$ and any input $\boldsymbol{x}$.
Suppose a SNIFF on target bias parameter $b_{S,j_0}$ was carried out.
Let $z_{j_0}$ and $\tilde{z}_{j_0}$ denote the correct and faulted value of $\mathcal{O}_{\theta,{j_0}}(\boldsymbol{x})$. 
Then the target weight $b_{S,j_0}$ can be recovered as:
\[
b_{S,j_0} = \frac{1}{2}\ln\left(\frac{\tilde{z}_{j_0}^{-1}-1}{z_{j_0}^{-1}-1}\right).
\]
\end{thm}
\begin{proof}
Let $j_0$ be given and let $\boldsymbol{x}$ be any input.
For any $j\in\{1,2,\dots,m\}$,
\[
y_{j} = b_{S,j} + \sum_{i=1}^n I_iw_{ij},\ 
z_{j} = \frac{1}{1 + \exp (-y_{j})}.
\]
In particular,
\[
y_{j_0} = b_{S,j_0} + \sum_{i=1}^n I_iw_{ij_0},\ 
z_{j_0} = \frac{1}{1 + \exp (-y_{j_0})}.
\]
Let
\begin{eqnarray*}
A &:=& \sum_{i=1}^n I_iw_{ij_0} = y_{j_0}-b_{S,j_0}.
\end{eqnarray*}
We have
\begin{eqnarray*}
z_{j_0}&=&\frac{1}{1 + \exp (-A-b_{S,j_0})},\\ \tilde{z}_{j_0}&=&\frac{1}{1 + \exp (-A+b_{S,j_0})}.
\end{eqnarray*}
We note that by definition of sigmoid, $z_{i_0}>0$ and $\tilde{z}_{i_0}>0$.
\begin{eqnarray*}
\frac{1}{z_{j_0}}-1&=&\exp (-A-b_{S,j_0})\\
\frac{1}{\tilde{z}_{j_0}}-1&=&\exp (-A+b_{S,j_0})
\end{eqnarray*}

By definition of sigmoid, $z_{j_0}^{-1}>1$,
\[
\frac{\tilde{z}_{j_0}^{-1}-1}{z_{j_0}^{-1}-1} =  \exp(2b_{S,j_0})\Longrightarrow b_{S,j_0} = \frac{1}{2}\ln\left(\frac{\tilde{z}_{j_0}^{-1}-1}{z_{j_0}^{-1}-1}\right).
\]
\end{proof}
\begin{cor}
The attacker can recover the bias vector $\boldsymbol{b}_S$ with $m$ faults and $2m$ executions of the target neural network (the student neural network). 
\end{cor}

\begin{thm}
For any $i_0\in\{1,2,\dots,n\}, j_0\in\{1,\dots,m\}$ and any $\boldsymbol{x}$, a non-vanishing input for $i_0$.
Suppose a SNIFF on target weight parameter $w_{i_0j_0}$ was carried out.
Let $z_{j_0}$ and $\tilde{z}_{j_0}$ denote the correct and faulted value of $\mathcal{O}_{\theta,{j_0}}(\boldsymbol{x})$. 
Then the target weight $w_{i_0j_0}$ can be recovered as:
\[
w_{i_0j_0} = \frac{1}{2I_{i_0}}\ln\left(\frac{\tilde{z}_{j_0}^{-1}-1}{z_{j_0}^{-1}-1}\right).
\]
\end{thm}
\begin{proof}
Let $i_0,j_0$ be given, and let $\boldsymbol{x}$ be a non-vanishing input for $i_0$.
%For simplicity, we write $\boldsymbol{I}$ (resp. $\boldsymbol{y}$) instead of $\boldsymbol{I}(\boldsymbol{x})$ (resp. $\boldsymbol{y}(\boldsymbol{x})$).
We let $w_{ij}$ denote the $(i,j)$th entry of the weight matrix $W_S$.
And let $b_{S,j}$ denote the $j$th entry of the bias vector $\boldsymbol{b}_S$.
Then for any $j\in\{1,2,\dots,m\}$,
\[
y_{j} = b_{S,j} + \sum_{i=1}^n I_iw_{ij},\ 
z_{j} = \frac{1}{1 + \exp (-y_{j})}.
\]
In particular,
\[
y_{j_0} = b_{S,j_0} + \sum_{i=1}^n I_iw_{ij_0},\ 
z_{j_0} = \frac{1}{1 + \exp (-y_{j_0})}.
\]
Let
\begin{eqnarray*}
A &:=& b_{S,j_0} + \sum_{i=1, i\neq {i_0}}^n I_iw_{ij_0} = y_{j_0}-I_{i_0}w_{i_0j_0}.
\end{eqnarray*}
We have
\begin{eqnarray*}
z_{j_0}&=&\frac{1}{1 + \exp (-A-I_{i_0}w_{i_0j_0})},\\ \tilde{z}_{j_0}&=&\frac{1}{1 + \exp (-A+I_{i_0}w_{i_0j_0})}.
\end{eqnarray*}
We note that by definition of sigmoid, $z_{i_0}>0$ and $\tilde{z}_{i_0}>0$.
\begin{eqnarray*}
\frac{1}{z_{j_0}}-1&=&\exp (-A-I_{i_0}w_{i_0j_0})\\
\frac{1}{\tilde{z}_{j_0}}-1&=&\exp (-A+I_{i_0}w_{i_0j_0})
\end{eqnarray*}
Since $\boldsymbol{x}$ is a non-vanishing input for $i_0$, we have $I_{i_0}\neq 0$.
Also by definition of sigmoid, $z_{i_0}^{-1}>1$.
Together with the above equations,
\[
\frac{\tilde{z}_{j_0}^{-1}-1}{z_{j_0}^{-1}-1} =  \exp(2I_{i_0}w_{i_0j_0})\Longrightarrow w_{i_0j_0} = \frac{1}{2I_{i_0}}\ln\left(\frac{\tilde{z}_{j_0}^{-1}-1}{z_{j_0}^{-1}-1}\right).
\]
\end{proof}
\begin{cor}
The attacker can recover the weight matrix $W_S$ with $mn$ faults and $2mn$ executions of the targeted neural network (the student neural network). 
\end{cor}
\subsubsection{Tanh}
In case the activation function for the last layer is tanh, for an input $\boldsymbol{x}$, the output of the student neural network is given by
\[
\mathcal{O}_{\theta}(\boldsymbol{x})=\text{tanh}(\mathcal{O}_{\theta_{T,-1}}(\boldsymbol{x})^TW_S+\boldsymbol{b}_S),
\]
%Let $\boldsymbol{y}(\boldsymbol{x}):=\mathcal{O}_{\theta_{T,-1}}(\boldsymbol{x})^TW_S+\boldsymbol{b}_S$, then we have 
for $i=1,2,\dots,m$,
\[
\mathcal{O}_{\theta,i}(\boldsymbol{x}) = \frac{\exp(y_i(\boldsymbol{x})) - \exp(-y_i(\boldsymbol{x}))}{\exp(y_i(\boldsymbol{x}))+\exp (-y_i(\boldsymbol{x}))}.
\]

\begin{thm}
For any $j_0\in\{1,2,\dots,m\}$ and any input $\boldsymbol{x}$.
Suppose a SNIFF on target bias parameter $b_{S,j_0}$ was carried out.
Let $z_{j_0}$ and $\tilde{z}_{j_0}$ denote the correct and faulted value of $\mathcal{O}_{\theta,{j_0}}(\boldsymbol{x})$. 
Then the target weight $b_{S,j_0}$ can be recovered as:
\[
b_{S,j_0} = \frac{1}{4}\ln\frac{(1+z_{j_0})(1-\tilde{z}_{j_0})}{(1-z_{j_0})(1+\tilde{z}_{j_0})}.
\]
\end{thm}
\begin{proof}
Let $j_0$ be given and let $\boldsymbol{x}$ be any input.
%For simplicity, we write $\boldsymbol{I}$ (resp. $\boldsymbol{y}$) instead of $\boldsymbol{I}(\boldsymbol{x})$ (resp. $\boldsymbol{y}(\boldsymbol{x})$).
For any $j\in\{1,2,\dots,m\}$,
\begin{eqnarray*}
y_{j} &=& b_{S,j} + \sum_{i=1}^n I_iw_{ij},\\ 
z_{j} &=& \frac{\exp (y_{j}) - \exp (-y_{j})}{\exp (y_{j}) + \exp (-y_{j})} = 1-\frac{2}{\exp(2y_{j}) + 1}.
\end{eqnarray*}
In particular,
\[
y_{j_0} = b_{S,j_0} + \sum_{i=1}^n I_iw_{ij_0},\ 
z_{j_0} = 1-\frac{2}{\exp(2y_{j_0}) + 1}.
\]
Let
\begin{eqnarray*}
A &:=& \sum_{i=1}^n I_iw_{ij_0} = y_{j_0}-b_{S,j_0}.
\end{eqnarray*}
We have
\begin{eqnarray*}
z_{j_0}&=&1-\frac{2}{\exp(2A+2b_{S,j_0}) + 1},\\ \tilde{z}_{j_0}&=&1-\frac{2}{\exp(2A-2b_{S,j_0}) + 1}.
\end{eqnarray*}
We note that by definition of tanh, $z_{i_0}<1$ and $\tilde{z}_{i_0}<1$.
\begin{eqnarray*}
\frac{1 + z_{j_0}}{1 - z_{j_0}}&=&\exp(2A+2b_{S,j_0})\\
\frac{1 + \tilde{z}_{j_0}}{1-\tilde{z}_{j_0}}&=&\exp(2A-2b_{S,j_0}),
\end{eqnarray*}
which gives
{\small
\[
\frac{(1+z_{j_0})(1-\tilde{z}_{j_0})}{(1-z_{j_0})(1+\tilde{z}_{j_0})}=\exp(4b_{S,j_0})\Longrightarrow b_{S,j_0} = \frac{1}{4}\ln\frac{(1+z_{j_0})(1-\tilde{z}_{j_0})}{(1-z_{j_0})(1+\tilde{z}_{j_0})}.
\]}
\end{proof}
\begin{cor}
The attacker can recover the bias vector $\boldsymbol{b}_S$ with $m$ faults and $2m$ executions of the target neural network (the student neural network). 
\end{cor}
\begin{thm}
For any $i_0\in\{1,2,\dots,n\}, j_0\in\{1,\dots,m\}$ and any $\boldsymbol{x}$, a non-vanishing input for $i_0$.
Suppose a SNIFF on target weight parameter $w_{i_0j_0}$ was carried out.
Let $z_{j_0}$ and $\tilde{z}_{j_0}$ denote the correct and faulted value of $\mathcal{O}_{\theta,{j_0}}(\boldsymbol{x})$. 
Then the target weight $w_{i_0j_0}$ can be recovered as:
\[
w_{i_0j_0} = \frac{1}{4I_{i_0}}\ln\frac{(1+z_{j_0})(1-\tilde{z}_{j_0})}{(1-z_{j_0})(1+\tilde{z}_{j_0})}.
\]
\end{thm}
\begin{proof}
Let $j_0$ be given and let $\boldsymbol{x}$ be any input.
%For simplicity, we write $\boldsymbol{I}$ (resp. $\boldsymbol{y}$) instead of $\boldsymbol{I}(\boldsymbol{x})$ (resp. $\boldsymbol{y}(\boldsymbol{x})$).
For any $j\in\{1,2,\dots,m\}$,
\[
y_{j} = b_{S,j} + \sum_{i=1}^n I_iw_{ij},\
z_{j} = 1-\frac{2}{\exp(2y_{j}) + 1}.
\]
In particular,
\[
y_{j_0} = b_{S,j_0} + \sum_{i=1}^n I_iw_{ij_0},\ 
z_{j_0} = 1-\frac{2}{\exp(2y_{j_0}) + 1}.
\]
Let
\begin{eqnarray*}
A &:=& b_{S,j_0} + \sum_{i=1, i\neq {i_0}}^n I_iw_{ij_0} = y_{j_0}-I_{i_0}w_{i_0j_0}.
\end{eqnarray*}
We have
\begin{eqnarray*}
z_{j_0}&=&1-\frac{2}{\exp(2A+2I_{i_0}w_{i_0j_0}) + 1},\\ \tilde{z}_{j_0}&=&1-\frac{2}{\exp(2A-2I_{i_0}w_{i_0j_0}) + 1}.
\end{eqnarray*}
We note that by definition of tanh, $z_{i_0}<1$ and $\tilde{z}_{i_0}<1$.
\begin{eqnarray*}
\frac{1 + z_{j_0}}{1 - z_{j_0}}&=&\exp(2A+2I_{i_0}w_{i_0j_0})\\
\frac{1 + \tilde{z}_{j_0}}{1-\tilde{z}_{j_0}}&=&\exp(2A-2I_{i_0}w_{i_0j_0}),
\end{eqnarray*}
which gives
{\small
\begin{eqnarray*}
\frac{(1+z_{j_0})(1-\tilde{z}_{j_0})}{(1-z_{j_0})(1+\tilde{z}_{j_0})}=\exp(4I_{i_0}w_{i_0j_0})\\
\Longrightarrow w_{i_0j_0} = \frac{1}{4I_{i_0}}\ln\frac{(1+z_{j_0})(1-\tilde{z}_{j_0})}{(1-z_{j_0})(1+\tilde{z}_{j_0})}.
\end{eqnarray*}}
\end{proof}
\begin{cor}
The attacker can recover the weight matrix $W_S$ with $mn$ faults and $2mn$ executions of the targeted neural network (the student neural network). 
\end{cor}
\subsubsection{Relu}
In case the activation function for the last layer is relu, for an input $\boldsymbol{x}$, the output of the student neural network is given by
\[
\mathcal{O}_{\theta}(\boldsymbol{x})=\max\{0,\mathcal{O}_{\theta_{T,-1}}(\boldsymbol{x})^TW_S+\boldsymbol{b}_S\},
\]
for $j=1,2,\dots,m$,
\[
\mathcal{O}_{\theta,j}(\boldsymbol{x}) = \max\{0,y_j(\boldsymbol{x})\}.
\]
When the output of the activation function is $0$, we cannot get much information.
Thus, for relu, we need to consider an additional faulting position, position 5 in Figure~\ref{fig:activation}.
The effect of the fault is to flip the sign of the result of the summation, i.e. $y_j$ with notation above.
Thus in case the output of relu was original zero, after fault, the output would be absolute value of the summation, i.e. $-y_j$.

Given any $i_0\in\{1,2,\dots,n\}, j_0\in\{1,\dots,m\}$ and any input $\boldsymbol{x}$.
%For the attack, the attacker first observe the output $\mathcal{O}_{\theta,{j_0}}(\boldsymbol{x})$.
Then there are two steps for the attack:

\noindent\textbf{Step 1:}
\begin{enumerate}
    \item If $\mathcal{O}_{\theta,{j_0}}(\boldsymbol{x})\neq0$.
    Let $z_{j_0}$ denote $\mathcal{O}_{\theta,{j_0}}(\boldsymbol{x})$.
    \item If $\mathcal{O}_{\theta,{j_0}}(\boldsymbol{x})=0$.
    The attacker executes the inference with the same input and inject SNIFF on the summation $y_i$.
    Let $z_{j_0}$ denote the negative of the faulted value of $\mathcal{O}_{\theta,{j_0}}(\boldsymbol{x})$.
\end{enumerate}
\textbf{Step 2:} The attacker executes the inference with the same input and inject SNIFF on target parameter - bias $b_{S,j_0}$ or weight $w_{i_0,j_0}$.
\begin{enumerate}
    \item If the faulted value of $\mathcal{O}_{\theta,{j_0}}(\boldsymbol{x})\neq0$.
    Let $\tilde{z}_{j_0}$ denote the faulted value of $\mathcal{O}_{\theta,{j_0}}(\boldsymbol{x})$.
    \item Otherwise, the attacker executes the inference with the same input and inject SNIFF on both the target parameter and the summation $y_i$.
    Let $\tilde{z}_{j_0}$ denote the negative of the faulted value of $\mathcal{O}_{\theta,{j_0}}(\boldsymbol{x})$.
\end{enumerate}
\begin{thm}
For any $j_0\in\{1,2,\dots,m\}$ and any input $\boldsymbol{x}$.
Following the attack steps described above, the target bias can be recovered as
\[
b_{S,j_0}=\frac{1}{2}(z_{j_0} - \tilde{z}_{j_0})
\]
\end{thm}
\begin{proof}
We note that with the above attack, we have
\[
z_{j_0} = y_{j_0} = b_{S,j_0} + \sum_{i=1}^n I_iw_{ij_0}, \tilde{z}_{j_0} = \tilde{y_{j_0}} = -b_{S,j_0} + \sum_{i=1}^n I_iw_{ij_0}.
\]
Thus
\[
b_{S,j_0} = \frac{1}{2}(z_{j_0} - \tilde{z}_{j_0})
\]
\end{proof}
\begin{cor}
The attacker can recover the bias vector $\boldsymbol{b}_S$ with at most $3m$ faults and at most $4m$ executions of the target neural network (the student neural network). 
\end{cor}
\begin{thm}
For any $i_0\in\{1,2,\dots,n\}, j_0\in\{1,\dots,m\}$ and any $\boldsymbol{x}$, a non-vanishing input for $i_0$.
Following the attack steps described above, the target weight $w_{i_0j_0}$ can be recovered as:
\[
w_{i_0j_0} = \frac{1}{2I_{i_0}}(z_{j_0} - \tilde{z}_{j_0}).
\]
\end{thm}
\begin{proof}
We note that with the above attack, we have
\begin{eqnarray*}
z_{j_0} &=& y_{j_0} = b_{S,j_0} + \sum_{i=1}^n I_iw_{ij_0},\\ \tilde{z}_{j_0} &=& \tilde{y_{j_0}} = b_{S,j_0} - I_{i_0}w_{i_0j_0}+ \sum_{i=1,i\neq i_0}^n I_iw_{ij_0}.
\end{eqnarray*}
Since $\boldsymbol{x}$ is a non-vanishing input for $i_0$, we have $I_{i_0}\neq 0$.
We have
\[
w_{i_0j_0} = \frac{1}{2I_{i_0}}(z_{j_0} - \tilde{z}_{j_0})
\]
\end{proof}
\begin{cor}
The attacker can recover the weight matrix $W_S$ with at most $3mn$ faults and at most $4mn$ executions of the targeted neural network (the student neural network). 
\end{cor}

\end{document}